\documentclass{article}
\usepackage{graphicx}
\usepackage{amsthm}
\usepackage{amssymb}
\usepackage{amsfonts}
\usepackage{chapterbib}
\usepackage{bbm}
\usepackage{amsmath}
\usepackage{cite}
\usepackage{color,fullpage,hyperref}

\newtheorem{prop}{Proposition}
\newtheorem{thm}{Theorem}
\newtheorem{lem}{Lemma}
\newtheorem{definition}{Definition}

\def\F{\mathcal{F}}
\def\A{\mathcal{A}}
\def\R{\mathbb{R}}
\def \sd{\Sigma\Delta}

\def\Cx{\frac{1}{\sqrt{\ell}}P_\ell V^* R_\Omega C_x}

\def\fn{\|\frac{1}{\sqrt{\ell}}P_\ell V^* R_\Omega C_x \|_F}
\def\2n{\|\frac{1}{\sqrt{\ell}}P_\ell V^* R_\Omega C_x\xi \|_2}
\def\E{\mathbb{E}}
\def\P{\mathbb{P}}
\def\D{D_{N,s}}
\def\l{\frac{1}{\sqrt{\ell}}}

\allowdisplaybreaks

\begin{document}

\title{Quantized Compressed Sensing\\ for Partial Random Circulant Matrices}
\author{Joe-Mei Feng, Felix Krahmer, Rayan Saab}

\maketitle

\abstract{
We provide the first analysis of a non-trivial quantization scheme for compressed sensing measurements arising from structured measurements. Specifically, our analysis studies compressed sensing matrices consisting of rows selected at random, without replacement, from a circulant matrix generated by a random subgaussian vector. We quantize the measurements using stable, possibly one-bit, Sigma-Delta schemes, and use a reconstruction method based on convex optimization. We show that the part of the reconstruction error due to quantization decays polynomially in the number of measurements. This is in-line with analogous results on Sigma-Delta quantization associated with random Gaussian or subgaussian matrices, and significantly better than results associated with the widely assumed memoryless scalar quantization. Moreover, we prove that our approach is stable and robust; i.e., the reconstruction error degrades gracefully in the presence of non-quantization noise and when the underlying signal is not strictly sparse.  
 The analysis relies on results concerning subgaussian chaos processes as well as a variation of McDiarmid's inequality. 
}

\section{Introduction}

Compressed sensing \cite{CRT05,CRT06,Donoho2006_CS} deals with accurately reconstructing sparse (or approximately sparse) vectors $x\in \R^N$ from relatively few generalized linear measurements of the form $(\langle  a_i , x \rangle)_{i=1}^m$, where $m<N$ and where the vectors $a_i\in \R^N$ are chosen appropriately.
Accurate reconstruction is  theoretically possible because well chosen compressed sensing measurement maps are injective on the ``low-complexity" set of sparse vectors. On the other hand, tractable reconstruction algorithms in the compressed sensing context rely heavily on sophisticated, non-linear techniques including convex optimization and greedy numerical methods (e.g., \cite{needell2009cosamp,dai2009subspace,blumensath2009iterative}). Consider the $m \times n $ matrix $A$ whose rows are given by the vectors $a_i$, and denote the possibly noisy compressed sensing measurements by \begin{equation}y=Ax+e,\label{eq:CS_model}\end{equation} where $e\in \R^m$ represents noise. If $\|e\|_2 \leq \epsilon$, and $A$ is chosen appropriately, then standard compressed sensing results guarantee (e.g., \cite{CRT05, CRT06, Donoho2006_CS}, see also \cite{foucart2013mathematical}) that the solution $\hat{x}$ to the optimization problem
\begin{equation}\label{eq:l1}
\min_z \|z\|_1 \text{\quad subject to \quad} \|Az-y\|_2 \leq \epsilon
\end{equation}
satisfies 
\begin{equation}\|x-\hat{x}\|_2 \leq C(\|e\|_2 + \frac{\|x-x_s\|_1}{\sqrt{s}} ).\label{eq:l1_guarantee}\end{equation}
Above, $x_s$ denotes the best $s$-sparse approximation to $x$ (i.e., the vector with at most $s$ non-zero entries that best approximates $x$).

 The need for sophisticated non-linear decoders such as \eqref{eq:l1}, which can only be reliably implemented on digital computers, implies that compressed sensing is inextricably linked to a digitization (quantization) step. Through quantization, the measurements are converted from continuous valued quantities to elements from a finite set (e.g., $\{\pm 1\}$), so that they can be stored and manipulated (and ultimately used for reconstruction) via digital computers.

Despite the importance of quantization, and a flurry of recent activity focusing on this subject in the compressed sensing context, its treatment remains rather underdeveloped in at least two ways. First, most of the current literature (e.g., \cite{sun2009optimal,zymnis2010compressed,jacques2011dequantizing,laska2011democracy,boufounos20081,plan2013one})  has focused on the most intuitive approach to quantization, namely memoryless scalar quantization (MSQ). However, MSQ is known to have strong theoretical limitations to its reconstruction error guarantees, which we discuss in Section \ref{sec:quant}. Second, all works on the topic to date have only considered compressed sensing matrices $A$ with subgaussian random entries, both for MSQ and for more sophisticated quantization schemes such as $\Sigma\Delta$ quantization, which have been shown to outperform MSQ (see Section~\ref{sec:quant} below for more details). 
\subsection{Contributions}In this paper, we address the lack of a non-trivial quantization theory for a practically important class of measurement matrices: partial random circulant matrices. Our main result, Theorem \ref{main} shows that if the compressed sensing measurement matrix is a randomly subsampled partial random circulant matrix, and the measurements are quantized by a stable (even $1$-bit) Sigma-Delta quantizer, then with an appropriate tractable decoder (which we specify):
\begin{itemize}
\item The reconstruction error due to quantization decays polynomially with the number of measurements.
\item The recovery is robust to noise  and stable with respect to deviations from the sparsity assumption.
\end{itemize}   
Our analysis relies on proving a restricted isometry property for the product of our compressed sensing measurement matrix and the matrix formed by the left singular vectors of an $r$th order difference operator, which we provide in Proposition \ref{inf}. 
 For this, we use a combination of a version of McDiarmid's inequality \cite{mcdiarmid1998concentration}, Dudley's inequality \cite{du67}, and recent results on suprema of chaos processes \cite{krahmer2014suprema}. As a notable technical difference to previous works (without quantization) studying measurement systems involving random subsampling, our proof explicitly exploits that we are subsampling without replacement.  Let us now introduce the necessary background information, starting with partial random circulant matrices, followed by a brief introduction to quantization and to the concentration of measure techniques we employ.
\section{Background and notation}
\subsection{Notation and basic definitions}
We denote by $[N]$ the set $\{1, \dots, N\}$ and by $e_{k}$ the $k$-th standard basis vector. A vector $x\in \R^N$ is $s$-sparse if only $s$ of its entries are non-vanishing, that is, its support $T=\operatorname{supp}(x) = \{j\in [N]: x_j \neq 0 \}$ satisfies $|T|=s$. Throughout, the matrix $F=\big(e^{2\pi ijk/N}\big)_{j,k=1}^N$ is the unnormalized $N\times N$ discrete Fourier transform matrix, and $\bar{F}$ denotes the complex conjugate of $F$. That is, $F\bar{F} = \bar{F}F=N Id$. We say that a matrix $A$ satisfies the restricted isometry property of order $s$ and constant $\delta$, if for all $s$-sparse vectors $x$ 
$$(1-\delta)\|x\|_2^2\leq  \|Ax\|_2^2 \leq (1+\delta)\|x\|_2^2.$$

Given a vector $x\in \R^N$, we denote by $\hat{X}\in \R^{N\times N}$ the diagonal matrix with $\hat{x}:=Fx$ on the diagonal. For a matrix $A$, $A_k$ denotes its $k$-th column.

We write $f \lesssim  g$ for two functions $f$ and $g$ if they are defined on the same domain $D$ and there exists an absolute constant $C$ such that $f(y)\leq C g(y)$ for all $y\in D$, $f \gtrsim  g$ is defined analogously.
Given a full-rank matrix $A \in \R^{m\times d}$ with $m>d$, its pseudo-inverse is given by $A^{\dagger}=(A^*A)^{-1}A^*$.

\subsection{Partial random circulant matrices}
Given a vector $\xi=(\xi_1,\xi_2,\ldots,\xi_N)\in \R^N$, the corresponding circulant matrix $\Phi=\Phi(\xi)\in \mathbb{R}^{N\times N}$
is defined by  

\begin{center}
\begin{equation}C_\xi=\left[
\begin{array}{ccccc}
\xi_1&\xi_{2}&\xi_{3}&\cdots&\xi_N\\
\xi_N&\xi_1&\xi_{2}&\cdots&\xi_{N-1}\\
\vdots&&&&\vdots\\
\xi_2&\xi_{3}&\xi_{4}&\cdots&\xi_1
\end{array}\right].\label{eq:Phi}\end{equation}
\end{center}
In this paper we consider random circulant matrices $C_\xi$ arising from random vectors $\xi$ whose entries are independent $L$-subgaussian random variables with variance 1 and mean 0, in the sense of the following definition. 
 \begin{definition}[see, e.g., \cite{vershynin2010introduction}]\label{def:SG}
 	A random variable $X$ is called $L$-subgaussian if 
 	\begin{align}
 	\P(|X|>t)  \leq \exp(1-t^2/L^2). \label{subgaussiantail}
 	\end{align}
 	Up to absolute multiplicative constants, the subgaussian parameter $L$ is equivalent to the subgaussian 
 	norm $\|X\|_{\Psi_2}$ defined as $\|X\|_{\Psi_2}=\sup_{p\geq 1}p^{-1/2}(\E|X|^p)^{1/p}$. Specifically,   \eqref{subgaussiantail} implies that \cite{vershynin2010introduction}
 	\begin{equation}\label{eq:willow}
 	\|X\|_{\psi_2}\leq \sqrt{\tfrac{e}{2}}L.
 	\end{equation} 
 \end{definition}

  A partial random circulant matrix is obtained from a random circulant matrix  by sampling the rows of the latter. In this paper, we consider only sampling without replacement, thus obtaining the following definition. 

\begin{definition}\label{def:rand_circ}
Let $\Phi=C_\xi \in \R^{N\times N}$ be a random circulant matrix as in \eqref{eq:Phi} and, for $m \leq N$, let $\Omega=(\Omega_1, \dots, \Omega_m)$ be a random vector obtained by sampling from $[N]$ without replacement. That is, $\Omega$ is drawn uniformly at random from the set 
\begin{equation}
\Xi := \{\omega\in [N]^m: \omega_i\neq\omega_j\mbox{ for }i\neq j\}.
\end{equation}

Then the associated \emph{partial random circulant matrix} is given by  \[
A = R_\Omega \Phi.
\]
where $R_\Omega$ is the subsampling operator
\[
 \R^{m\times N} \ni R_\Omega = \sum_{j=1}^m e_{j} e_{\Omega_j}^*.
 \]

\end{definition}

Partial random circulant matrices are important to the practical application of compressed sensing. This is due to the simple observation that a circular convolution of a signal $x\in \R^N$ with a ``filter" $\widetilde{\xi}\in \R^N$, as given by the vector
$y=x \circledast \tilde \xi \in \R^N$ with entries \[y_j:= \sum_{i=1}^{N} x_i \widetilde{\xi}_{j-i \bmod n},\] 
can be represented by the action of a circulant matrix. Indeed one has $x\circledast \tilde \xi = C_\xi x$, where 

 $\xi \in \R^N$ is defined via $ \xi_{N-j+1}=\widetilde{\xi}_j$ for $j\in \{1,...,N\}$ and $C_\xi$ is as in \eqref{eq:Phi}.
Consequently, as the convolution is commutative, one has $C_\xi x=C_x \xi$; we will repeatedly make use of this observation.

 Due to the ubiquity of convolutions in signal processing applications, partial random circulant matrices, modeling subsampled random convolutions, have played an important role in the development of compressed sensing applications such as radar imaging, Fourier optical imaging, and wireless channel estimation (see, e.g., \cite{romberg2009compressive,haupt2010toeplitz}). Recovery guarantees for partial circulant matrices have been an active area of research in the last decade, the best known results have recently been proved by Mendelson, Rauhut, and Ward \cite{MRW17}.

\subsection{Quantization}\label{sec:quant}
In the compressed sensing context, quantization is the map that replaces the vector $y = Ax + e \in \R^m$ by a representation that uses a finite number of bits. Most often, practical quantization maps are of the form 
\begin{align}
\mathcal{Q}: \R^m &\to \mathcal{A}^m \nonumber
\end{align}
where $\mathcal{A}\subset \R$ is a finite set, called the quantization alphabet.
Both memoryless scalar quantization and $\Sigma\Delta$ quantization, which we will discuss in the next paragraphs, execute quantization maps of this form.  

The most natural and common choices of alphabets have equispaced elements. As representatives for such alphabets we will focus on the so-called mid-rise alphabet with $2L$ levels and step-size $\delta$, denoted by $\A_L^\delta$ and given by 
$\mathcal{A}_L^\delta:=\big\{\pm(2\ell+1)\delta/2, \ \ell \in\{0,...,L-1\} \big\}.\label{eq:mid-rise}$
The minimal instance of such an alphabet is the $1$-bit quantization alphabet, which we denote by $\mathcal{A}=\{-1,+1\}$.

  The fact that $\mathcal{Q}$ outputs a vector of alphabet elements allows the quantization to be implemented progressively. That is, one can relate each entry of the quantized vector to some measurement and each subsequent measurement can then be quantized in a way that depends on previous measurements. This idea is exploited in $\sd$ schemes.

\subsection*{Memoryless scalar quantization}
Memoryless scalar quantization is an intuitive approach to digitizing compressed sensing measurement. It simply uses a scalar quantizer 
\begin{align}\label{eq:scalar}
Q_\A: \R &\to \mathcal{A} \nonumber \\
      z &\mapsto \arg\min_{v\in \mathcal{A}} | z-v |  
\end{align}
to quantize every entry of $y$ independently.
Using a standard compressed sensing recovery algorithm such as \eqref{eq:l1}, one can use the robustness of standard compressed sensing reconstruction algorithms \eqref{eq:l1_guarantee}  to bound the reconstruction error. Such results guarantee that the reconstruction error decays as the size of the alphabet increases. However, they  {do not guarantee error decay as one takes more measurements}. One could argue that a better reconstruction algorithm or a sharper analysis would alleviate this issue, but that is hardly the case. Indeed, consider working with a fixed quantization alphabet, as one would do in practice due to fixing the quantization hardware. Then, as shown by Goyal, Vetterli, and Thao \cite{GVT98}, the error in reconstructing a $k$-sparse signal from its $m$ MSQ-quantized measurements cannot decay faster than $k/m$, {even when using an optimal decoder}. This means that by linearly increasing the number of measurements, and hence increase the number of bits used, denoted by $\mathcal{R}$ (for rate), one can, at best, only linearly decrease the reconstruction error, denoted $\mathcal{D_{MSQ}}$ (for distortion). That is, the rate-distortion relationship associated with MSQ satisfies  \begin{equation}\mathcal{D}_{MSQ}(\mathcal{R}) \geq C \mathcal{R}^{-1}.\label{eq:RD_MSQ}\end{equation} This lower bound stands in sharp contrast to the rate-distortion relationship that an optimal assignment of bits (for encoding $k$-sparse vectors in the unit-ball of $\R^N$) yields, namely (see, e.g., \cite{BB_DCC07}) $$\mathcal{D}^*(\mathcal{R})\leq C\frac{N}{k}e^{-c \mathcal{R}/k}.$$ {In this sense MSQ is far from optimal.}
One factor preventing MSQ from being optimal in general, is that it does not exploit any correlations among the measurements, as it treats each measurement independently of the others.

\subsection*{Sigma-Delta quantization}

Sigma-Delta ($\Sigma\Delta$) quantization
 is an alternative quantization method that, in its simplest form, works by scalar quantizing the sum of the current measurement and a state variable, and then updating the state variable. 
It is through the state variable that the dependencies between the measurements are accounted for in the quantization. $\sd$ schemes were proposed in the 1960's \cite{inose1963unity} for quantizing bandlimited functions and have seen widespread use in practice, particularly in audio applications  \cite{NST96}. For almost 40 years, there was no precise understanding of $\sd$ from a mathematical perspective, before recently, following the seminal work of Daubechies and Devore in \cite{daub-dev}, a number of works analyzed $\sd$ schemes for bandlimited functions from a mathematical perspective \cite{G03, DGK10, KW12, DS15}. 

 In addition, $\sd$ schemes have recently been shown to be well suited for quantizing finite-frame expansions \cite{benedetto2006sigma, blum:sdf,BPA2007,KSW12} as well as compressed sensing measurements \cite{gunturk2013sobolev, KSY13,SWY16, SWY16_2}. We review these results in the following subsection, and we now focus on the relevant details of $\Sigma\Delta$ quantization schemes.  

In the simplest $\sd$ scheme, a first order $\sd$ quantizer, the state variable $u_i$ accounts for the accumulated quantization error.
That is, the quantizer applies to the measurements $y_i$
the iteration
\begin{align}\label{eq:SDrec}
q_i&=Q_\mathcal{A}(y_i+u_{i-1}) \\
u_i&=u_{i-1}+y_i-q_i \label{eq:SDrec2}.
\end{align}
Here $Q_\mathcal{A}$ is the scalar quantizer \eqref{eq:scalar}. In an  {$r^{\rm th}$-order $\sd$ scheme, the first order finite difference $\Delta$ (given by $(\Delta u)_i := u_i - u_{i-1}$, and appearing in \eqref{eq:SDrec2}) is replaced by an $r^{\rm th}$-order finite difference $\Delta^r$. Moreover,  before applying the scalar quantizer, some quantization rule $\rho:  \mathbb{R}^{r+1} \rightarrow \mathbb{R}$ is applied.

That is, the quantized measurement vector $q$ with entries $q_i \in \A$ is computed via the recursion 
\begin{equation}
q_i = Q_\A\left(\rho(y_i,u_{i-1},u_{i-2},\dots,u_{i-r})\right),
\label{equ:rthOrdq}
\end{equation}
\begin{equation}
u_i =  y_i - q_i - \sum^{r}_{j=1} {r \choose j} (-1)^j u_{i-j}
\label{equ:rthOrdu}. 
\end{equation}
Using the first-order difference matrix $D$ with entries given by \begin{equation}
D_{i,j} := \left\{ \begin{array}{ll} 1 & \textrm{if}~i=j\\ -1 & \textrm{if}~i = j+1\\ 0 & \textrm{otherwise}\end{array} \right.,
\label{Def:diffMat}
\end{equation}
the relationship between ${ x}$, ${ u}$, and ${ q}$ can be concisely written in matrix-vector notation as 
\begin{equation}
D^r { u} ~=~ y - { q}.
\label{equ:LinrSigDelt}
\end{equation} 
The inverse $D^{-r}$ will play a crucial role in our analysis, which is why we fix the notation $$D^{-r}=USV^*,$$ for its singular value decomposition throughout this paper.

Recalling that $(D^{-1}z)_j = \sum_{i=1}^j z_i$, in the case of  first order schemes (where $r=1$) the state variable $u$ can be interpreted as an accumulated error, as can be seen by appying $D^{-1}$ to the equation above. It intuitively follows that it is crucial for the sequence of state variable $u$ to be bounded in this case. This intuition can be made precise and generalizes to higher order schemes. For this reason we seek  \textit{stable $r^{\rm th}$-order schemes}, i.e., schemes for which \eqref{equ:rthOrdq} and \eqref{equ:rthOrdu} result in $$\|  u \|_{\infty} \leq {C}_{\rho,Q}(r)$$ for all $N \in \mathbb{N}$, and ${ y} \in \mathbb{R}^N$ with $\|{ y}\|_\infty \leq 1$. Importantly, we require that $C_{\rho,Q}:  \mathbb{N} \mapsto \mathbb{R}^+$ be entirely independent of both $N$ and ${ y}$.  One can show that stable $r^{\rm th}$-order $\Sigma \Delta$ schemes exist with $C_{\rho,Q}(r) = O((Cr)^r)$ for some constant $C$ \cite{G03, DGK10}, even when $\mathcal{A}$ is a $1$-bit alphabet, but that there are fundamental lower bounds on $C$ and no better dependence on $r$ can be achieved \cite{CD02, KW12}. 

\subsection{Probabilistic Tools}\label{sec:prob_tools}
We will use a number of different probabilistic tools for different parts of our argument. We state them here for convenience.
The first one is a  variation of McDiarmid's inequality. Note that it closely relates to the Azuma-Hoeffding inequality and the method of bounded differences.
\begin{thm} [\cite{mcdiarmid1998concentration}, Theorem 3.14]\label{DMcDiarmid} Let $(\Omega, \mathcal{F}, \P)$ be a probability space and $(\emptyset, \Omega) = \mathcal{F}_0 \subseteq  \mathcal{F}_1 \subseteq ... \subseteq  \mathcal{F}_m  $ a filtration in $\mathcal{F}$. Consider a bounded random variable $X$, and set $X_k := \E(X|\mathcal{F}_k)$.
	Define the sum of squared conditional ranges $$R^2 = \sum_{k=1}^{m} ran_k^2$$  where $$ran_k := \sup(X_k|\F_{k-1}) + \sup(-X_k| \F_{k-1}),$$
	and denote its (essential) supremum by $$\hat{r}^2:=\sup R^2.$$
	Then, \[ \P(X-\E(X)\geq t) \leq e^{-2t^2/\hat{r}^2}.\]
\end{thm}

A second tool that we will be using is Dudley's inequality. In order to formulate the result, we recall the definitions of the covering number and of subgaussian random variables.
\begin{definition}
	Let $(S,d)$ be a metric space and $\epsilon>0$. A subset $\mathcal{N}_\epsilon$ of $S$ is called an $\epsilon$-net if every point in $S$ can be approximated to within $\epsilon$ by some point in $\mathcal{N}_\epsilon$, i.e., for all $x\in S$  there exists $y\in \mathcal{N}_\epsilon$ such that $d(x,y)<\epsilon$.  
	The {\em covering number} $\mathcal{N}(S,d,\epsilon)$ is the minimal cardinality of an $\epsilon$-net of $S$.
\end{definition}

\begin{thm}[Dudley's inequality \cite{du67}]\label{chaining}
	Let $Z_x$ be a random variable depending on $x\in T$, for some set $T$ and
	define $d(x,y)=\|Z_x-Z_y\|_{\Psi_2}$, if
	$$
	\P(|Z_x-Z_y|>t)\lesssim \exp\Big(-t^2/\|Z_x-Z_y\|^2_{\Psi_2}\Big),
	$$
	then for any $x_0\in T$
	$$
	\P(\sup    | Z_x-Z_{x_0}    |>t)\lesssim \exp\Big(-t^2/\big(\int_0^{\sup_{x\in \D} \|Z_x\|_{\Psi_2} } \sqrt{\log \mathcal{N}  (\D,d(x,y), \epsilon)}d\epsilon\big)^2\Big).
	$$
\end{thm}
A third result that we will be using concerns subgaussian chaos processes.
Its original version involves the Talagrand $\gamma_2$ functional, an intricate complexity parameter related to the generic chaining \cite{talagrand2014upper}, which can be bounded in terms of covering numbers via Dudley's inequality (Theorem~\ref{chaining}). To avoid discussing the generic chaining methodology in detail, we state a combined version in terms of only these upper bounds.

\begin{thm}[\cite{krahmer2014suprema}]\label{chaos}
Let $\mathcal{C}$ be a set of matrices and consider the complexity parameters
\begin{align*}
d_F(\mathcal{C})&=\sup_{C\in\mathcal{C}}\|C\|_F,\qquad
d_{2\rightarrow 2}(\mathcal{C})  =\sup_{C\in\mathcal{C}}   \|C\|_{2\rightarrow 2},\qquad
D(\mathcal{C})=\int_0^{d_{2\rightarrow 2}(\mathcal{C})}\sqrt{\log\mathcal{N}(\mathcal{C},\|\cdot\|_{2\rightarrow 2},u)}\ du.
\end{align*}
Let $\xi$ be a random vector whose entries $\xi_j$ are independent, mean-zero, variance $1$, $L$-subgaussian random variables. Then, for $t>0$, the random variable 
\[C_\mathcal{C}(\xi) =\sup_{C\in \mathcal{C}}  |     \| C\xi\|_2^2-\E_{\xi}\|C\xi\|_2^2 \]
satisfies
$$
\P(C_\mathcal{C}(\xi)  \geq c_1 E+t  )\leq 2 \exp(-c_2\min\{\frac{t^2}{V^2},\frac{t}{U}\}),
$$
where
\begin{align*}
E&=D(\mathcal{C})(D(\mathcal{C})+d_F(\mathcal{C}))+d_F(\mathcal{C})d_{2\rightarrow 2}(\mathcal{C}),\qquad
V=d_{2\rightarrow 2}(\mathcal{C})(D(\mathcal{C})+d_F(\mathcal{C})),\qquad
U=d_{2\rightarrow 2}^2(\mathcal{C}),
\end{align*}
and the constants $c_1,\ c_2$ depend only on $L$.
\end{thm}

\section{Related Work}
\subsection{$\Sigma\Delta$ quantization of finite-frame expansions}
The first paper analyzing $\Sigma\Delta$ quantization of finitely many measurements of finite-dimensional vectors was  
\cite{benedetto2006sigma}, initiating a series of papers on the subject. For example, the papers  \cite{benedetto2006sigma,  BPA2007, blum:sdf, KSW12} all studied $\Sigma\Delta$ quantization when one collects $m>N$ linear measurements $y_i= \langle a_i, x \rangle$ of $x\in \R^N$, where the collection $(a_i)_{i=1}^m$ spans $\R^N$ (and is called a finite-frame). In this \emph{finite-frame} setting, \cite{benedetto2006sigma} showed that the reconstruction error associated with first order $\Sigma\Delta$ quantization can be made to decay linearly with the number of measurements, hence the bit-rate. With this first order $\Sigma\Delta$ approach, the upper bound on the error  already matched the lower bound \eqref{eq:RD_MSQ} associated with MSQ. Using higher order $\Sigma\Delta$ schemes, subsequent papers (e.g., \cite{BPA2007, blum:sdf}) showed that the error can be made polynomial in the number of measurements, significantly outperforming the MSQ lower bound. Importantly, the linear reconstruction scheme proposed in \cite{BPA2007} to approximate $x$ from its quantized finite-frame measurements also proved fruitful in the compressed sensing context. Denoting by $A$ the $m\times N$ matrix ($m\geq N$) having $a_i$ as its rows, the \emph{$r^{\rm th}$-order Sobolev dual} of $A$ is the $N\times m$ matrix $$B:= (D^{-r}A)^\dagger D^{-r},$$
which is easily seen to be a left-inverse of $A$. The approach of \cite{BPA2007} was to estimate $x$ from $q$ via $\hat{x}=Bq$, yielding error rates that decayed like $m^{-r}$ (i.e., polynomially in the number of measurements and bits) provided the rows of $A$ obeyed some smoothness conditions. 

\subsection{$\Sigma\Delta$ quantization of compressed sensing measurements}
The paper \cite{gunturk2013sobolev}, soon followed by \cite{KSY13, feng2014rip}, was first to study $\Sigma\Delta$ quantization of compressed sensing measurements. They focused on the setup where the compressed sensing matrix is subgaussian, the underlying signal is \emph{strictly sparse}, and no noise contaminates the measurements. They analyzed a two-stage approach to signal recovery whereby one uses a standard decoder like \eqref{eq:l1} to estimate the support of the $k$-sparse signal, then applies the Sobolev dual of the associated $m\times k$ sub-matrix of $A$ to $q$. With this approach, the reconstruction error was again shown to decay polynomially in the number of measurements. The proofs in \cite{gunturk2013sobolev, KSY13} relied on bounding the smallest singular value of a certain anisotropic random matrix, while \cite{feng2014rip} significantly simplified the analysis by using an approach based on the restricted isometry property. These results showed that frame-theoretic quantization techniques could be extended to the compressed sensing setup. On the other hand, the reliance of \cite{gunturk2013sobolev, KSY13} on a two-step approach involving support recovery meant that obtaining a result for compressed sensing measurements of \emph{arbitrary} signals in the presence of noise would be difficult. 

More recently, in \cite{SWY15} a decoder based on convex optimization was proposed (to replace the two-step approach) and analyzed, with the main result being that it could handle both arbitrary signals and measurement noise (bounded by $\epsilon$). Specifically, if $q$ results from  quantizing compressed sensing measurements $y$ (as in \eqref{eq:CS_model}) using an $r^{\rm th}$-order $\Sigma\Delta$ scheme, one approximates $x$ with $\hat{x}$ via
\begin{align}\label{eq:opt}
(\hat{x},\hat{e}) :=  \arg\min\limits_{(z,\nu)}\|z\|_1 \  \text{ subject to } &  \|D^{-r}(A z+\nu-q )\|_2 \leq \gamma(r)\sqrt{m}  \notag \\
\text{\ \ and\ \ } & \|\nu\|_2\leq \epsilon \sqrt m,
\end{align} 
where $\gamma(r)$ depends on the quantization scheme used. The resulting approximation error due to quantization in \cite{SWY15} decays as $m^{-r+1/2}$, i.e., polynomially in $m$, and the approach is shown to be stable and robust. As in \cite{feng2014rip}, a main ingredient in the proofs of \cite{SWY15} is an analysis based on the restricted isometry properties of certain matrices arising from the interaction of the difference matrix with the compressed sensing matrix. Indeed, the following result, which we will also use, is proved in \cite{SWY15}. 
\begin{thm}\label{thm:SWY}\cite{SWY15}
Let $ A$ be an $m\times N$ matrix, and let $k,l \in \{1,...,m\}$. Suppose that $\frac{1}{\sqrt{\ell}}{P_\ell V^* A}$ satisfies the restricted isometry property of order $2k$ and constant $\delta<1/9$. Denote by $Q_{\Sigma\Delta}^r$ a stable $r$th order $\Sigma\Delta$ quantizer. Then, for all $x\in \mathbb{R}^N$ with $\| A x \|_\infty \leq \mu <1$ and all $e \in \mathbb{R}^m$ with $\|e\|_\infty \leq \epsilon < 1-\mu$ the estimate $\hat{x}$ obtained by solving \eqref{eq:opt} with $q=Q_{\Sigma\Delta}^r(A x + e)$
satisfies 
\begin{equation}\label{eq:guarantee}
\|\hat{x}-x\|_2 \leq C_1 \left(\frac{m}{\ell}\right)^{-r+1/2}\delta+C_2 \frac{\sigma_k(x)}{\sqrt k}+C_3\sqrt{\frac{m}{\ell}}\epsilon,
 \end{equation}
 where the constants $C_1,C_2,C_3$ depend on the quantizer, but not the dimensions of the problem.
\end{thm}

The combination of stability, robustness, quantization error decay, and practicability make the $\Sigma\Delta$ quantization approach, followed by recovery via \eqref{eq:opt}  amenable to practical applications where one has the freedom to select subgaussian compressed sensing matrices. Nevertheless, the only matrices $\Phi$ for which \cite{SWY15} proved that the assumptions of Theorem \ref{thm:SWY} hold are subgaussian. As such, the results of \cite{SWY15} do not apply to important practical setups such as system identification, radar, and coded-aperture imaging, where \emph{structured} random matrices such as partial random circulant ones arise naturally in the compressed sensing context (see, e.g., \cite{haupt2010toeplitz, rauhut2012restricted}). The only result we are aware of (aside from those of this manuscript) that addresses quantization in the context of  structured random measurement matrices is that of \cite{Wang15}. \cite{Wang15} shows that first order $\Sigma\Delta$ quantization coupled with an appropriate decoder yields an error decaying as $\big(\frac{m}{k^4\log N }\big)^{-1/2}$, when the measurement matrix is a randomly selected $m \times N$  submatrix of the $N\times N$ discrete Fourier transform matrix. Consequently the results are only meaningful when $m$ scales like $k^4$, which is considerably worse than the linear scaling of $m$ with $k$ (up to log factors) arising in Theroem \ref{thm:SWY} and commonly in compressed sensing without quantization. 
One of our main contributions  (Theorem \ref{main}) is to show that such a linear scaling (up to log factors) 
also holds for certain structured random measurements, specifically for random circulant matrices.

\section{Main results}\label{sec:results}

In this section, we prove the following theorem, which is the main result of this paper.
\begin{thm}\label{main} Denote by $Q_{\Sigma\Delta}^r$ a stable $r$th order $\Sigma\Delta$ quantizer.  Let $A$ be an $m\times N$ partial random circulant  matrix associated to a vector with independent $L$-subgaussian entries with mean 0 and variance 1. 
 Suppose that $N\geq m \geq (C \eta)^\frac{1}{1-2\alpha} s  \log^{\frac{2}{1-2\alpha}} N\log^{\frac{2}{1-2\alpha}}s $, for some $\eta > 1$ and $\alpha\in[0,1/2)$. With probability exceeding $1-e^{-\eta}$, the following holds:
\newline 
For all $x\in \mathbb{R}^N$ with $\| A x \|_\infty \leq \mu <1$ and all $e \in \mathbb{R}^m$ with $\|e\|_\infty \leq \epsilon < 1-\mu$ the estimate $\hat{x}$ obtained by solving \eqref{eq:opt} satisfies
\[\|\hat{x}-x\|_2 \leq C_1 \left(\frac{m}{\ell}\right)^{-r+1/2}\delta+C_2 \frac{\sigma_k(x)}{\sqrt k}+C_3\sqrt{\frac{m}{\ell}}\epsilon.\]
Here $C,C_1,C_2, C_3$ are constants that only depend on $r$ and $L$.

\end{thm}

\begin{proof}Theorem \ref{main} can be immediately obtained from Theorem \ref{thm:SWY}, which requires a bound on the restricted isometry constants of  $ P_\ell V^* R_{\Omega} C_\xi$ where $\ell=m(\frac{s}{m})^\alpha$, and Proposition \ref{inf} below, which provides the required bound. 
\end{proof}
\begin{prop}\label{inf}Consider the same setup and assumptions as Theorem \ref{main}; in particular assume that $m \geq ( C\eta)^\frac{1}{1-2\alpha} s  \log^{\frac{2}{1-2\alpha}} N\log^{\frac{2}{1-2\alpha}}s $, for some $\eta > 1$ and $\alpha\in[0,1/2)$. Setting $\ell=m(\frac{s}{m})^\alpha$, we have
\begin{align*}
&\P\Big(\sup_{x} |\2n^2-1|>\frac{1}{9}
\Big)
 <
e^{-\eta},
\end{align*} where the supremum is over all $s$-sparse vectors. In other words, with probability exceeding $1-e^{-\eta}$, the matrix
 $\frac{1}{\sqrt{\ell}}P_\ell V^*R_\Omega C_\xi$ satisfies the restricted isometry property of order $s$, with constant $1/9$.
\end{prop}
\begin{proof} Note that by the triangle inequality,
\begin{align}\label{eq:triangle}
&\sup_x \Big| \|   \l  P_\ell V^* R_{\Omega} C_x \xi  \|_2^2-1 \Big|\notag\\
&\leq \sup_x \Big(\Big|  \2n^2 -\E[\2n^2|\Omega]\Big|+\notag\\
&\quad \Big| \E[ \2n^2  |\Omega]-\E\2n^2 \Big|+\notag\\
&|\mathbb{E} \|\frac{1}{\sqrt{\ell}}P_\ell V^* R_\Omega C_x \xi\|_2^2-1|\Big).
\end{align}
Thus, the proof of Proposition \ref{inf} boils down to 
controlling each of the summands in \eqref{eq:triangle}.
To that end, Lemma \ref{lowerbound} (below) shows that the third summand is bounded by $\frac{sm}{\ell N}$, while Lemma  \ref{part1} and Lemma \ref{part2} bound the probability that the remaining summands exceed $\frac{1}{18}$ and $\frac{1}{36}$ respectively. Our bound on $m$ (potentially with an increased value of $C$) ensures that $\frac{sm}{\ell N} \leq \frac{s}{\ell} = \left(\frac{s}{m}\right)^{1-\alpha}\leq \frac{1}{36}$  
and the result follows using a union bound. 
\end{proof}

\begin{lem}\label{lowerbound}
Given the same setup as in Theorem \ref{main} and Proposition \ref{inf}, one has
$$|\mathbb{E} \|\frac{1}{\sqrt{\ell}}P_\ell V^* R_\Omega C_x \xi\|_2^2-1| \leq \frac{(s-1)(m-\ell)}{\ell(N-1)}\leq \frac{sm}{\ell N}.$$
\end{lem}
\begin{proof}
Denoting by $c_{i,j}$ the $(i,j)$-th entry of $C_x$ 
 and noting that we are sampling without replacement, we observe that for $p \neq q \in [m]$
\begin{align}
\E ( c_{\Omega(p),k} c_{\Omega(q),k})
&=\frac{1}{N(N-1)}\sum_{u\neq v=1}^N c_{u,k}c_{v,k}
=\frac{1}{N(N-1)}\Big(\sum_{u,v=1}^N c_{u,k}c_{v,k}-\sum_{u=1}^N c^2_{u,k}\Big)\nonumber\\
&=\frac{1}{N(N-1)}\Big(\sum_{u,v=1}^N c_{u,k}c_{v,k}-\sum_{u=1}^N x_u^2\Big)
=\frac{1}{N(N-1)} \Big(\big(\sum_{u=1}^N x_{u}\big)^2-1\Big)\label{noreplacement}.
\end{align}
The last two equalities both use the fact that each row of $C_x$ is a shifted copy of $x$. Furthermore 
\begin{align*}
\Big|\mathbb{E} \| \frac{1}{\sqrt{\ell}}P_{\ell}V^* R_{\Omega}C_x \xi\|_{2}^2-1\Big|
&=\Big|\mathbb{E}\|\frac{1}{\sqrt{\ell}}P_{\ell}V^* R_\Omega C_x\|_F^2-1\Big|\nonumber\\
&=\Big|\frac{1}{\ell}\mathbb{E} \sum_{j=1}^{\ell}\sum_{k=1}^N|\sum_{p=1}^mv_{jp}c_{\Omega(p),k}|^2-1\Big|\nonumber\\
&=\Big|\frac{1}{\ell}\sum_{j=1}^\ell\sum_{k=1}^N \big(\sum_{p=1}^m v_{jp}^2\mathbb{E}c^2_{\Omega(p),k}
+\sum^m_{p,q=1\atop p\neq q}v_{jp}v_{jq}\mathbb{E} c_{\Omega(p),k}c_{\Omega(q),k}\big)-1\Big|\nonumber\\
&=\Big|\frac{1}{\ell}\sum_{j=1}^{\ell}\big(1+ \frac{
(\sum_{i=1}^N x_i)^2-1}{N-1}               \sum^m_{p,q=1\atop p\neq q}v_{jp}v_{jq}  \big)-1\Big|\nonumber\\
\end{align*}
where in the last equality we used \eqref{noreplacement} and the fact that the rows of both $C_x$ and $V$ are normalized. Using that $x$ is $s$-sparse, it follows that
\begin{align*}
\Big|\mathbb{E} \| \frac{1}{\sqrt{\ell}}P_{\ell}V^* R_{\Omega}C_x \xi\|_{2}^2-1\Big| &\leq \Big|
  \frac{s-1}{ \ell(N-1)}           
  \Big( \sum_{j=1}^{\ell}  (\sum_{p=1}^m v_{jp}      )^2 - \sum_{j=1}^{\ell}  \sum_{p=1}^m v_{jp}      ^2   \Big)        \Big|
\\&= \frac{s-1}{\ell (N-1)}         
\Big|\|V^* (1,\ldots,1)^T\|_2^2  - \ell         \Big|\\
&\leq \frac{s-1}{\ell (N-1)} 
\Big|\|V \|_{2\rightarrow 2}^2 m - \ell          \Big|\\
&=\frac{(s-1)(m-\ell)}{\ell(N-1)}.
\end{align*}
\end{proof}

\begin{lem}\label{part1}
Consider again the setup of Theorem \ref{main} and Proposition \ref{inf} and denote by $D_{N,s}$ the set of all $s$-sparse vectors in $\mathbb{R}^N$. Then 

\begin{align*}
&\P\bigg(\sup_{x\in\D}\bigg|\2n^2- \E_{\xi}\Big[\2n^2 \Big|\Omega\Big]   \bigg|>\frac{1}{18}\Big|\bigg)
\leq\frac{1}{2} e^{-\eta}
.\end{align*}
\end{lem}
\begin{proof}

We will apply Theorem \ref{chaos} conditionally given $\Omega$ with  
$\mathcal{C}=\{\Cx:x\in\D\}$. This set is almost the same as the one considered in the proof of Theorem 4.1 in \cite{krahmer2014suprema}, the only differences being the additional projection $P_\ell$ and our normalization factor of $\l$ (instead of ${\frac{1}{\sqrt{m}}}$ in \cite{krahmer2014suprema}). Indeed, since $\|P_\ell\|_{2\rightarrow 2}\leq 1$ we can estimate the necessary parameters for applying Theorem \ref{chaos} exactly as in the proof of Theorem 4.1 in \cite{krahmer2014suprema}. This yields
\begin{align*}
&d_{2\rightarrow 2}(\mathcal{C})
\leq \sqrt{\frac{s}{\ell}}, \qquad
d_F(\mathcal{C})
\leq \sqrt{\frac{m}{\ell}}, \qquad
D(\mathcal{C})
\leq\sqrt{\frac{s}{\ell}}\log N\log s.
\end{align*}
Consequently for $c_1$, $c_2$, and $E$ as in Theorem \ref{chaos}, we have
\begin{align*}
E&\leq \sqrt{\frac{s}{\ell}}\log N\log s\left(\sqrt{\frac{s}{\ell}}\log N\log s+\sqrt{\frac{m}{\ell }}\right)+\sqrt{\frac{m}{\ell }}\sqrt{\frac{s}{\ell}}\\
 &\leq \left(\frac{s}{m}\right)^{1-\alpha}\log^2N\log^2s + 2\left(\frac{s}{m}\right)^{1-2\alpha}\log N\log s
 \leq \frac{1}{36c_1}.
 \end{align*}
Here, the second inequality follows from our choice of $\ell$ and the last inequality follows from our assumption on $m$ in Theorem \ref{main} (potentially adjusting the constant $C$). Again adjusting the constant, we similarly obtain 
 \begin{align*}
V&\leq \sqrt{\frac{c_2}{4\eta}} \qquad \text{ and } \qquad U\leq\frac{c_2}{4\eta}.
\end{align*}
Hence the probability is bounded by $2e^{-4\eta}$. Finally, as $\eta\geq 1$, $e^{-4\eta}\leq \frac{1}{4}e^{-\eta}$ and the result follows by taking the expectation over $\Omega$.

\end{proof}

\begin{lem}\label{part2} With the same notation as before, we have
\begin{align*}
&\P(\sup_{x\in\D}|  \E[\2n^2|\Omega]-\E\2n^2   |>\frac{1}{36})\\
&\leq C' \exp(   -c /(\frac{\sqrt{sm}}{\ell}\log N\log {m})^2) \leq \frac{1}{2} e^{-\eta}
\end{align*}
where $c$, $C'$ are constants that depends only on $L$.
\end{lem}

\begin{proof}
The proof is a direct application of Theorem \ref{chaining} for the random variable 
\[Z_x:=\E  \bigg[\2n^2-\E\2n^2\bigg|\Omega\bigg] =  \fn^2-\E\fn^2\] 
 to find the supremum of the deviation. Since Theorem \ref{chaining} requires the covering number with respect to the metric $d(x,y):=\|Z_x-Z_y\|_{\Psi_2}$ we need a bound for $d(x,y)$, which we provide in Lemma~\ref{distance} below. 
Specifically, the first inequality in Lemma \ref{part2} follows from Theorem \ref{chaining} together with Lemma \ref{lowerbound} and Lemma \ref{part1} above. 
Indeed, applying Lemma \ref{distance} with $y=0$ yields
\begin{align}
\sup_{x,y}\|Z_x\|_{\Psi_2}\leq \frac{\sqrt{m}}{\ell}\|x\|_{\hat{\infty}}\leq  \frac{\sqrt{m}}{\ell}\|F(x)\|_{\infty}\leq
 \frac{\sqrt{m}}{\ell}\|x\|_{1}\leq   \frac{\sqrt{sm}}{\ell}\|x\|_{2}\leq \frac{\sqrt{sm}}{\ell}.
\end{align}
To bound the integral in Theorem \ref{chaining}, we note that $$\mathcal{N}(\D, \frac{\sqrt{m}}{\ell}\|\cdot\|_{\hat{\infty}},\epsilon)=\mathcal{N}(\D, \frac{1}{\sqrt{m}}\|\cdot\|_{\hat{\infty}},\frac{\ell}{m}\epsilon),$$
and hence applying the argument in \cite[Section 4]{krahmer2014suprema} scaled by $\frac{m}{\ell}$,
 \begin{align*}
 & \int_0^{\sup_x \|Z_x\|_{\Psi_2} } \sqrt{\log \mathcal{N}  (\D,{\frac{1}{\sqrt{m}}}\|\cdot\|_{\hat{\infty}}, \frac{\ell}{m}\epsilon)}d\epsilon\\
&\lesssim \frac{\sqrt{sm}}{\ell} \log  N\log s. 
 \end{align*}

 For the second inequality note that by the definition of $\ell$ and the assumed lower bound on $m$
\begin{align}(\frac{\sqrt{sm}}{\ell}\log N\log {s})^2 &= \left(\frac{s}{m}\right)^{1-2\alpha}\log^2 N \log^2 s \\
& \leq C^{-1}\eta^{-1}.
\end{align}
The result follows from the assumption that $\eta \geq 1$ as in the proof of Lemma \ref{part1}.
\end{proof}

All that remains now is to prove Lemma \ref{distance}. Before that, we derive a technical bound required for its proof.

\begin{lem}\label{difference}
Let $\omega$, $\omega ' \in \Xi= \{\omega\in [N]^m: \omega_i\neq\omega_j\mbox{ for }i\neq j\}$ be such that $\omega$ differs from $\omega'$ in at most two components. Then the function 
$$f(\omega):=  \|   \frac{1}{\sqrt{   \ell   }}      P_\ell       V^*       R_\omega       C_x \|_F^2   -\|   \frac{1}{\sqrt{   \ell   }}      P_\ell       V^*       R_\omega       C_y \|_F^2$$
 satisfies 
$$    |     f(\omega)-f(\omega ')|\leq \frac{    12}{     \ell}\|     x-y     \|_{\hat{     \infty}}, $$
where $\|x\|_{\hat{\infty}}:=\|Fx\|_{\infty}$.
\end{lem}
\begin{proof}
Note that, as a circulant matrix is diagonalized by the Fourier transform,
\begin{align}
f(\omega)&=  \|   \frac{1}{\sqrt{   \ell   }}      P_\ell       V^*       R_\omega       C_x \|_F^2   -\|   \frac{1}{\sqrt{   \ell   }}      P_\ell       V^*       R_\omega       C_y \|_F^2\nonumber\\
&=\|\frac{1}{\sqrt{\ell}}P_\ell V^* R_{\omega}F^{-1}\hat{X}F\|_F^2-\|\frac{1}{\sqrt{\ell}}P_\ell V^* R_{\omega}F^{-1}\hat{Y}F\|_F^2\nonumber\\
& = \frac{   1}{   \ell   N } \|    P_\ell V^*  R_\omega \overline{F}    \hat{X}     \|_F^2- \frac{   1}{   \ell   N  } \|    P_\ell V^* R_\omega  \overline{F}    \hat{Y}     \|_F^2\nonumber  \\
& = \frac{  1  }{     \ell   N }  \sum  _{  k=1 }^  N \big(  |  \hat{ x}_k   |^2     -
  |  \hat{ y}_k   |^2  \big)    \|    P_\ell   V^*       R_{       \omega       }       \overline{       F       }_k         \| _2^2,
  \label{eq:mangrove}
\end{align}
where $F$ denotes the non-normalized Fourier transform, $F_k^T$  its $k$-th row, and $\hat x =Fx$. 

We first consider the case that $\omega$ and $\omega'$ differ only in one component, say the first (without loss of generality).  
To bound $|f(\omega)-f(\omega')|$ for this case, we note that for $V_j^T$ denoting the $j$-th row of $V$, and $\eta =\exp(-\tfrac{2\pi i}{N})$ an $N$-th root of unity, 
\begin{align*}
&\| P_\ell V R_\omega \overline{F}_k     \|^2_2-\|       P_\ell V R_{\omega '}  \overline{F}_k\|_2^2\\
&=\sum_{p,q=1}^m  \langle   \eta^{-k\omega_p} P_\ell V_p     ,\eta^{-k \omega_q} P_\ell V_q \rangle -
\sum_{r,s=1}^m \langle  \eta^{-k\omega_r '} P_\ell V_r     ,\eta^{-k \omega_s '} P_\ell V_s\rangle\\
&=\sum_{p,q=1}^m  (   \eta^{  k(\omega_p-\omega_q)     }   -  \eta^{   k(\omega_p '-\omega_q ')    }           )  \langle P_\ell V_p, P_\ell  V_q\rangle \\
&= {(  \eta^{k(\omega_1-\omega_1)   } -  \eta^{  k  (   \omega_1 '-\omega_1 '        )         }     )}  \langle P_\ell V_1, P_\ell  V_1\rangle +
\sum_{q=2}^m{(  \eta^{k(\omega_1-\omega_q)   } -  \eta^{  k  (   \omega_1 '-\omega_q         )         }     )}   \langle P_\ell V_1, P_\ell  V_q\rangle \\
&\quad+ \sum_{p=2}^m{(  \eta^{k(\omega_p-\omega_1)   } -  \eta^{  k  (   \omega_p -\omega_1 '        )         }     )}   \langle P_\ell V_p, P_\ell  V^*_1\rangle +
\sum_{p,q=2}^m ( \eta^{k    (  \omega_p -\omega_q )}-     \eta^{    k  (  \omega_p -\omega _q   )            }   )\langle P_\ell V_p, P_\ell  V_q\rangle \\
&=\sum_{q=2}^m  ( \eta^{k\omega_1} -  \eta^{ k\omega_1  '   }    )  \eta^{-k\omega_q}\langle P_\ell V_1, P_\ell  V_q\rangle +
\sum_{p=2}^m  ( \eta^{-k\omega_1} -  \eta^{ -k\omega_1  '   }    )  \eta^{k\omega_p}\langle P_\ell V_p, P_\ell  V_1\rangle.
\end{align*}

Combining this with \eqref{eq:mangrove}, we obtain
\begin{align}
    f(\omega)      -f(      \omega ')      
 &=\frac{1}{\ell N} \sum_{k=1} ^N   \big( |   \hat{x}_k  |^2-    |   \hat{y}_k  |^2 \big)    \Big( 
 \sum_{q=2}^m  (\eta^{k\omega_1}  -  \eta^{k\omega_1'} ) \eta^{-k \omega_q}\langle  P_\ell V_1,  P_\ell  V_q    \rangle \nonumber \\ &\qquad + 
\sum_{p=2}^m  (\eta^{-k\omega_1}  -  \eta^{-k\omega_1'} ) \eta^{k \omega_p}\langle  P_\ell V_p,  P_\ell  V_1    \rangle 
  \Big)\label{summands} 
\end{align}
Observe that  the right hand side is a sum of four different rescaled Fourier coefficients of the vector $u\in\mathbb{R}^N$ given by $u_k:=|\hat{x}_k|^2-|\hat{y}_k|^2$, as for example
\begin{align*}
   \frac{1}{\ell N} \sum_{p=2}^m    \langle  P_\ell V_1,  P_\ell  V_p   \rangle   \sum_{k=1}^N  (    |   \hat{x}_k   |^2-| \hat{y}_k  |^2    )     \eta^{ k( \omega_p -\omega_1 )   }  
&= \frac{1}{\ell N} \sum_{p=2}^m \langle  P_\ell V_1,  P_\ell  V_p   \rangle    (\overline{F}  u)_{\omega_p-\omega_1} = V_1^*P_\ell^* P_\ell V^* v,
\end{align*}
where $v\in \R^{m}$ is given by $v_1=0$ and $v_p= (\overline{F}u)_{\omega_{p}-\omega_1}$ for $2\leq p\leq m$. Note that as $\omega\in \Xi$ and hence the $\omega_q$ are all different, $v$ is a projection of $\overline{F}u$ on a subset of its entries, and so $\|v\|_2 \leq \sqrt{N}\|u\|_2$. Note that in this step, it is crucial to sample without replacement, as otherwise, the bound would no longer hold.
Consequently, using the Cauchy-Schwartz inequality,

\begin{align*}
\tfrac{1}{\ell N}\Big|  \sum_{p=2}^m    \langle  P_\ell V_1,  P_\ell  V_p   \rangle  \! \sum_{k=1}^N  \big(|   \hat{x}_k   |^2-| \hat{y}_k  |^2    \big)     \eta^{k( \omega_p -\omega_1)   } \Big|&\leq  \tfrac{1}{\ell N} \|V\|_{2\rightarrow 2}^2\|P_\ell^*P_\ell V_1\|_2\|  v   \|_2 \\ &\leq \tfrac{1}{\ell \sqrt{N}} \|\overline{F}u\|_2 \leq \tfrac{1}{\ell}\| \overline{F}u   \|_{\infty}= \tfrac{1}{\ell}  \|x-y\|_{\hat{\infty}}.
\end{align*}
Identical bounds for the other three summands in (\ref{summands})  are attained in an analogous way, which yields the result for $\omega$ and $\omega'$ differing in only one component (with a constant of $4$ rather than $12$).
 If they differ in two components, replacing one of these components in both $\omega$ and $\omega'$ by an entry that appears in neither of them, yields $\omega'', \omega''' \in\Xi$, which differ only in the other one of these components. Thus applying the above bound three times yields the result.
\end{proof}

We are now ready to bound the distance $d(x,y)= \|x-y\|_{\psi_2}$.
\begin{lem}\label{distance}  For all $x,y \in \R^N$ it holds that 
$$d(x,y)              \leq            \frac    {12 \sqrt{m}}    {\ell}      \|x-y\|_{\hat{\infty}}.$$
\end{lem} 
\begin{proof}
By \eqref{eq:willow}, it suffices to show that for all $t\geq 0$
\begin{equation}
\P_{\Omega}(|Z_x-Z_y|>t)\leq \exp\Big(1-t^2/\big(\frac{12\sqrt{m}}{\ell}\|x-y\|_{\hat{\infty}}\big)^2\Big).
\label{eq:spanishmoss}
\end{equation}

To prove this, we will apply Theorem \ref{DMcDiarmid} for $\F_k$, the $\sigma$-algebra generated by $\Omega_1,...,\Omega_k$.
For that, we need to bound the sum of squared ranges 
\begin{align*}
R^2=\sup\sum_{j=1}^m ran_j^2 
\end{align*}
where, for $(\Omega'_j,...,\Omega'_m)$  an independent copy of $(\Omega_j,...,\Omega_m)$ and $\Omega'=(\Omega_1,...\Omega_{j-1},\Omega'_j,...,\Omega'_m)$,
\begin{align}ran_j&=\sup_{\Omega_j \notin \{\Omega_1,...,\Omega_{j-1}\}}\Big(\E(f(\Omega)|\Omega_{j},...,\Omega_1) \Big| \Omega_{j-1},...,\Omega_1\Big)+\sup_{\Omega_j \notin \{\Omega_1,...,\Omega_{j-1}\}}\Big(\E(-f(\Omega')|\Omega_{j},...,\Omega_1) \Big| \Omega_{j-1},...,\Omega_1\Big) \nonumber \\
&=\sup_{\Omega_j,\Omega'_j \notin \{\Omega_1,...,\Omega_{j-1}\}}\Big(\E(f(\Omega)|\Omega_{j},\Omega_{j-1}...,\Omega_1) +\E(-f(\Omega')|\Omega'_{j},\Omega_{j-1},...,\Omega_1) \Big| \Omega_{j-1},...,\Omega_1\Big).
\label{eq:ranj}
\end{align}

For that, define the events $\mathcal{E}_0 = \{ \Omega_j\neq \Omega'_k\ \forall j>k \}$, $\mathcal{E}_0'=\{ \Omega'_j\neq \Omega_k\ \forall j>k \}$, and, for $j\in [m-k]$, $\mathcal{E}_j = \{ \Omega_{k+j}= \Omega'_k \}$, $\mathcal{E}'_j = \{ \Omega'_{k+j} = \Omega_k \}$ and note that 
\begin{equation}
\P[\cup_{j=0}^m \mathcal{E}_j| \Omega_1,...,\Omega_k,\Omega_k']= \P[\cup_{j=0}^m \mathcal{E}'_j| \Omega_1,...,\Omega_k,\Omega_k'] =1. \label{eq:olivetree}
\end{equation}
 Now, we can write
\begin{align}
\E[f(\Omega)|\Omega_1,...,\Omega_{k-1},\Omega_k]-\E[f(\Omega')|\Omega_1,...,\Omega_{k-1},\Omega'_k]
=& \sum_{j=0}^{m-k} 
\E[f(\Omega)\mathbbm{1}_{\mathcal{E}_j} -f(\Omega')\mathbbm{1}_{\mathcal{E}'_j} | \Omega_1,...,\Omega_k,\Omega_k'].
\label{eq:venusflytrap}
\end{align}
Given $\Omega_1, \dots, \Omega_k$ and $\Omega'_k$, consider random variables $\Omega''_{k+1}, \dots \Omega''_m$  drawn subsequently without replacement from $[N]\setminus\{\Omega_1,\dots, \Omega_k, \Omega'_k\}$ and set 
\[\Omega''=(\Omega_1, \dots, \Omega_k, \Omega''_{k+1},\dots, \Omega''_{m}), \qquad \Omega'''=(\Omega_1, \dots, \Omega_{k-1}, \Omega'_k, \Omega''_{k+1},\dots, \Omega''_{m}).
\]

 Observe that given the event $\mathcal{E}_0$, $\Omega$ and $\Omega''$ are conditionally identically distributed, and the same holds for $\Omega'$ and $\Omega'''$ given the event $\mathcal{E}'_0$.
 So, using that $\mathcal{E}_0$ and $\Omega''$ as well as $\mathcal{E}'_0$ and $\Omega'''$ are conditionally independent given $\Omega_1, \dots, \Omega_k, \Omega_k'$,  the summand in \eqref{eq:venusflytrap} corresponding to $j=0$ becomes 
\begin{align}
\E[f(\Omega)\mathbbm{1}_{\mathcal{E}_0} &-f(\Omega')\mathbbm{1}_{\mathcal{E}'_0} | \Omega_1,...,\Omega_k,\Omega_k']\nonumber \\=& \E[f(\Omega'')\mathbbm{1}_{\mathcal{E}_0} -f(\Omega''')\mathbbm{1}_{\mathcal{E}'_0} | \Omega_1,...,\Omega_k,\Omega_k']\nonumber \\=& 
\E[f(\Omega'')  | \Omega_1,...,\Omega_k,\Omega_k'] \P[\mathcal{E}_0 | \Omega_1,...,\Omega_k,\Omega_k']  -\E[f(\Omega''')| \Omega_1,...,\Omega_k,\Omega_k'] \P[\mathcal{E}'_0 | \Omega_1,...,\Omega_k,\Omega_k'] \label{eq:figtree} \\=& 
\big(\E[f(\Omega'') -f(\Omega''')| \Omega_1,...,\Omega_k,\Omega_k'] \big)\P[\mathcal{E}_0 | \Omega_1,...,\Omega_k,\Omega_k'],\nonumber\\
\leq & \frac{     12}{     \ell}\|     x-y     \|_{\hat{     \infty}} \P[\mathcal{E}_0 | \Omega_1,...,\Omega_k,\Omega_k'] \nonumber
\end{align}
where the third equality uses that $(\Omega'_k, \dots, \Omega'_m)$ is an independent copy of $(\Omega_k, \dots, \Omega_m)$ and so the two probabilities in \eqref{eq:figtree} are equal. The last inequality holds almost surely and follows from Lemma~\ref{difference}.

To bound the summand in \eqref{eq:venusflytrap} for $j>0$, we proceed in a similar way. Given $\Omega_1, \dots, \Omega_k$ and $\Omega'_k$, consider random variables $\Omega''_{k+1}, \dots, \Omega''_{k+j-1}, \Omega''_{k+j+1}, \dots\Omega''_m$  drawn subsequently without replacement from $[N]\setminus\{\Omega_1,\dots, \Omega_k, \Omega'_k\}$ and set 
\begin{align*}
\Omega''=&(\Omega_1, \dots, \Omega_k, \Omega''_{k+1},\dots,  \Omega''_{k+j-1}, \Omega_k', \Omega''_{k+j-1} ,\dots, \Omega''_{m}), \\
\Omega'''=&(\Omega_1, \dots, \Omega_k', \Omega''_{k+1},\dots,  \Omega''_{k+j-1}, \Omega_k, \Omega''_{k+j-1} ,\dots, \Omega''_{m}).
\end{align*}
As before, observe that given the event $\mathcal{E}_j$, $\Omega$ and $\Omega''$ are conditionally identically distributed, and the same holds for $\Omega'$ and $\Omega'''$ given the event $\mathcal{E}'_j$. The remainder of the estimate proceeds exactly as for $j=0$, with the slight difference that $\Omega''$ and $\Omega'''$ now differ in two entries, but nevertheless Lemma~\ref{difference} still applies. Thus we obtain
\begin{align}
\E[f(\Omega)\mathbbm{1}_{\mathcal{E}_j} -f(\Omega')\mathbbm{1}_{\mathcal{E}'_j} | \Omega_1,...,\Omega_k,\Omega_k']
\leq \frac{     12}{     \ell}\|     x-y     \|_{\hat{     \infty}} \P[\mathcal{E}_j | \Omega_1,...,\Omega_k,\Omega_k']. \nonumber
\end{align}
Consequently, one has almost surely
\begin{align*}
\E[f(\Omega)|\Omega_1,...,\Omega_{k-1},\Omega_k]-\E[f(\Omega')|\Omega_1,...,\Omega_{k-1},\Omega'_k] \leq& \sum_{j=0}^{m-k}  \frac{     12}{     \ell}\|     x-y     \|_{\hat{     \infty}} \P[\mathcal{E}_j | \Omega_1,...,\Omega_k,\Omega_k'] \\
=&  \frac{     12}{     \ell}\|     x-y     \|_{\hat{     \infty}},
\end{align*}
where the last equality follows from \eqref{eq:olivetree}, and hence, by \eqref{eq:ranj}, $ran_j \leq  \frac{     12}{     \ell}\|     x-y     \|_{\hat{     \infty}}$ and $R^2 \leq \big( \frac{  12 \sqrt{m}}{     \ell}\|     x-y     \|_{\hat{     \infty}}\big)^2$.
 
With this bound, Theorem \ref{DMcDiarmid} can be applied. One obtains
$$
\P(|Z_x-Z_y|>t)\leq 2\exp( -t^2/(     \frac{12\sqrt{m}}{\ell} \|x-y\|_{\hat{\infty}}     )^2  ),
$$
which implies \eqref{eq:spanishmoss}. We conclude
$$
d(x,y):=\|Z_x-Z_y\|_{\Psi_2}\leq \frac{12 \sqrt{m}}{\ell} \|x-y\|_{\hat{\infty}},
$$
as desired.
\end{proof}

\section*{Acknowledgements}
The three authors acknowledge support by the Hausdorff Institute for Mathematics (HIM), where part of this work was completed in the context of the HIM Trimester Program "Mathematics of Signal Processing", and support by the German Science Foundation in the context of the Emmy Noether Junior Research Group ``Randomized Sensing and Quantization of Signals and Images'' (KR 4512/1-1). In addition, JF and FK acknowledge support by the German Science Foundation in the context of the Research Training Group 1023 ``Identification in Mathematical Models". RS acknowledges support by a Hellman Fellowship and the NSF under DMS-1517204.

\bibliographystyle{plain}
\bibliography{QCSref}
\end{document}